\documentclass[pdflatex,sn-mathphys-num]{sn-jnl}

\usepackage{graphicx}
\usepackage{multirow}
\usepackage{amsmath,amssymb,amsfonts}
\usepackage{amsthm}
\usepackage{mathrsfs}
\usepackage[title]{appendix}
\usepackage{xcolor}
\usepackage{textcomp}
\usepackage{manyfoot}
\usepackage{booktabs}
\usepackage{algorithm}
\usepackage{algorithmicx}
\usepackage{algpseudocode}
\usepackage{listings}
\usepackage{braket}
\usepackage{amsfonts}
\usepackage{bm}

\theoremstyle{thmstyleone}
\newtheorem{lemma}{Lemma}
\newtheorem{proposition}{Proposition}

\newtheorem{corollary}{Corollary}
\newtheorem*{example*}{Example}
\newtheorem*{remark*}{Remark}
\newtheorem*{remarks*}{Remarks}

\newcommand{\dd}{\,\mathrm{d}}
\newcommand{\tr}{\,\mathrm{tr}}
\newcommand{\ee}{\,\mathrm{e}}

\begin{document}

\title[BKM ensemble]{Average entropy of Bogoliubov-Kubo-Mori random state ensemble}

\author{\fnm{Sohail}\footnotemark[1]~~\!\footnotemark[2]~~and \fnm{Lu} \sur{Wei}\footnotemark[2]
\footnotetext[1]{Institute for Quantum Studies, Chapman University, California 92866, USA}
\footnotetext[2]{Department of Computer Science, Texas Tech University, Texas 79409, USA}
\vspace{0.9cm}}

\abstract{Random states play a foundational role in different branches of modern quantum science. In this work, we study a recently proposed random state ensemble induced from von Neumann entropy through the Bogoliubov-Kubo-Mori (BKM) metric. In particular, we derive an exact yet explicit formula of average entanglement entropy over BKM ensemble. In obtaining the formula, we only make use of properties of normalization constant of the ensemble in the absence of its correlation kernel, contrary to average entropy computation of other ensembles. This new framework paves the way for calculating higher-order cumulants of BKM ensemble beyond the average.}

\keywords{Bogoliubov-Kubo-Mori metric; entropy-induced ensemble; random matrices; quantum entanglement; special functions; von Neumann entropy}

\maketitle

\section{Introduction and main results}\label{sec:main}
Crucial to successful exploitation of revolutionary advances of quantum science is the understanding of quantum entanglement. Entanglement is the physical phenomenon, the medium, and, most importantly, the resource underpinning quantum technologies. Mathematical description of entanglement is naturally formulated in terms of generic states, which are random states generated from information-geometric metrics. Random states find various applications in contemporary quantum science such as complexity of quantum circuits~\cite{Brandao21}, benchmark of quantum devices~\cite{Choi23}, and estimation of quantum entanglement~\cite{Page93}.

For an $m\times m$ density matrix $\rho$ of a bipartite system having the spectrum $0\leq\lambda_{m}<\dots<\lambda_{1}\leq1$ with the constraint
\begin{equation}\label{eq:rho}
\tr(\rho)=1,
\end{equation}
the degree of entanglement as measured by von Neumann entropy
\begin{equation}\label{eq:vN}
S=-\tr\left(\rho\ln\rho\right)=-\sum_{i=1}^{m}\lambda_{i}\ln\lambda_{i}
\end{equation}
has been a subject of intense study for several random state ensembles including Hilbert–Schmidt (HS) ensemble~\cite{Page93,Foong94,Ruiz95,VPO16,Wei17,Wei20,HWC21,HW25}, Bures–Hall (BH) ensemble~\cite{Sommers03,Forrester16,Kumar19,Wei20BHA,Wei20BH,WHW25,WHW26}, and fermionic Gaussian ensemble~\cite{Bianchi21,HW22,HW23}. The fact that different metrics induce those ensembles, while all utilizing the von Neumann entropy~(\ref{eq:vN}), thereby raises the question: what is the most natural random state ensemble in the space of density matrices when using von Neumann entropy as entanglement measure?

An answer to this question was recently proposed in~\cite{Miller25} by introducing what was termed as an entropy-based ensemble. Specifically, the von Neumann entropy~(\ref{eq:vN}) was shown to induce the Bogoliubov-Kubo-Mori (BKM) metric~\cite{Balian86,Balian14}
\begin{equation}
\dd^{2}S=-\sum_{i,j=1}^{m}\frac{\ln\lambda_{i}-\ln\lambda_{j}}{\lambda_{i}-\lambda_{j}}\dd\rho_{i,j}\dd\rho_{j,i}
\end{equation}
that gives rise to the BKM random state ensemble~\cite{Miller25}
\begin{equation}\label{eq:fBKM}
f(\lambda)=\frac{\Gamma(\beta)}{C(\alpha)}~\delta\left(1-\sum_{i=1}^{m}\lambda_{i}\right)\prod_{1\leq i<j\leq m}\left(\lambda_{i}-\lambda_{j}\right)\left(\ln\lambda_{i}-\ln\lambda_{j}\right)\prod_{i=1}^{m}\lambda_{i}^{\alpha},
\end{equation}
where
\begin{equation}\label{eq:b}
\beta=m\left(\alpha+\frac{m+1}{2}\right)
\end{equation}
and the normalization constant $C(\alpha)$ is shown in Lemma~\ref{l:C} to be
\begin{equation}\label{eq:C}
C(\alpha)=\Gamma^{m}(\alpha+1)\prod_{i=1}^{m}\Gamma(i).
\end{equation}
The Dirac delta function $\delta(\cdot)$ in~(\ref{eq:fBKM}) ensures that the ensemble is a proper density matrix satisfying~(\ref{eq:rho}). When von Neumann entropy is utilized as an entanglement measure, the BKM ensemble is therefore the most suitable choice than other known random state ensembles. Note that the BKM ensemble proposed in~\cite{Miller25} is in fact special case of~(\ref{eq:fBKM}) when $\alpha=-1/2$. Since generalization to the ensemble~(\ref{eq:fBKM}) of an arbitrary real $\alpha>-1$ is straightforward, we will refer to~(\ref{eq:fBKM}) as the BKM ensemble.

A first step to study a generic state ensemble is its the typical behavior of entanglement entropy over the ensemble. It was derived in~\cite{Miller25} the following asymptotic average entropy of the BKM ensemble~(\ref{eq:fBKM}),
\begin{equation}\label{eq:asy}
\mathbb{E}\left[S\right]=\ln m-\gamma-\ln2+\frac{1}{2}+\mathcal{O}\left(\frac{\ln m}{m}\right),
\end{equation}
where $\gamma\approx0.5772$ is the Euler's constant. The main result of this work, summarized in Corollary~\ref{main:m} below, is an exact yet explicit average entropy formula of the BKM ensemble. To obtain the result, we first compute the average of spectral moments
\begin{equation}\label{eq:Rk}
R_{l}=\sum_{i=1}^{m}x_{i}^{l}
\end{equation}
of the unconstrained BKM ensemble~(\ref{eq:BKM}) as presented in Proposition~\ref{main:k}. Instead of direct calculation, the route of computing average entropy through its spectral moments has recently been carried out over different ensembles~\cite{HW25,WHW25}. This new route is promising in efficiently finding higher-order moments of entropy beyond average values.
\begin{proposition}\label{main:k}
The average spectral moments~(\ref{eq:Rk}) of the unconstrained BKM ensemble~(\ref{eq:BKM}) valid for $l\geq0$ is given by
\begin{equation}\label{eq:k}
\mathbb{E}\left[R_{l}\right]=\sum_{i=0}^{m-1}\binom{m}{i+1}\frac{l^{i}}{i!}\left(\frac{\Gamma(\alpha+l+1)}{\Gamma(\alpha+1)}\right)^{(i)},
\end{equation}
where $\binom{m}{i+1}$ is the binomial coefficient.
\end{proposition}
The proof of Proposition~\ref{main:k} is found in Section~\ref{sec:P1}. As in~(\ref{eq:k}), we will use superscript in a parentheses $(i)$ to denote $i$-th derivative with respect to $\alpha$.

As a consequence of Proposition~\ref{main:k}, we now present the main result of this work in the following corollary.
\begin{corollary}\label{main:m}
The average von Neumann entropy~(\ref{eq:vN}) of the BKM ensemble~(\ref{eq:fBKM}) is given by
\begin{equation}\label{eq:m}
\mathbb{E}\left[S\right]=\psi_{0}(\beta+1)-\frac{1}{\beta}\left(\sum_{i=0}^{m-1}\binom{m}{i+1}\frac{1}{i!}\left(\alpha+\frac{m+1}{i+2}\right)\psi_{i}(\alpha+1)+\frac{1}{2}m(m+1)\right).
\end{equation}
\end{corollary}
The proof of Corollary~\ref{main:m} is in Section~\ref{sec:C1}. The function $\psi_{i}(x)$ in~(\ref{eq:m}) is the $i$-th order polygamma function~\cite{Brychkov}
\begin{equation}\label{eq:polygamma}
\psi_{i}(x)=\frac{\dd^{i+1}}{\dd x^{i+1}}\ln\Gamma(x)=\frac{\dd^{i}}{\dd x^{i}}\psi_{0}(x).
\end{equation}

Contrary to the computation of average entropy over other ensembles~\cite{Page93,Foong94,Ruiz95,Kumar19,Wei20BHA,Bianchi21}, the derivation of~(\ref{eq:k}) solely relies on the density~(\ref{eq:BKM}) of BKM ensemble, without using its correlation kernel or associated orthogonal polynomials. In fact, explicit expressions of these objects of the BKM ensemble are unavailable in the literature with the exception of a double-contour representation of the correlation kernel in~\cite{Forrester17}. Nevertheless, it remains to be seen if an efficient method can be developed to derive higher-order moments of entropy using only eigenvalue densities. A natural first step is to rederive mean entropy formulas in~\cite{Page93,Foong94,Ruiz95,Kumar19,Wei20BHA,Bianchi21} following a similar procedure of this work.

As examples, we list in below special cases of the main result~(\ref{eq:m}) for the first a few $m$. \\

$\blacktriangleright~m=2$:
\begin{equation}
\mathbb{E}\left[S\right]=-\frac{\alpha+1}{2\alpha+3}\psi_{1}(\alpha+1)-\psi_{0}(\alpha+1)+\psi_{0}(2\alpha+4)-\frac{3}{2\alpha+3}
\end{equation}

$\blacktriangleright~m=3$:
\begin{eqnarray}
\mathbb{E}\left[S\right]&=&-\frac{\alpha+1}{6(\alpha+2)}\psi_{2}(\alpha+1)-\frac{3\alpha+4}{3(\alpha+2)}\psi_{1}(\alpha+1)-\psi_{0}(\alpha+1)\nonumber\\
&&+~\!\psi_{0}(3\alpha+7)-\frac{2}{\alpha+2}
\end{eqnarray}

$\blacktriangleright~m=4$:
\begin{eqnarray}
\mathbb{E}\left[S\right]&=&-\frac{\alpha+1}{12(2\alpha+5)}\psi_{3}(\alpha+1)-\frac{4\alpha+5}{4(2\alpha+5)}\psi_{2}(\alpha+1)-\frac{3\alpha+5}{2\alpha+5}\psi_{1}(\alpha+1)\nonumber\\
&&-~\!\psi_{0}(\alpha+1)+\psi_{0}(4\alpha+11)-\frac{5}{2\alpha+5}
\end{eqnarray}
Compared to the leading-order expression~(\ref{eq:asy}) that may reflect the asymptotic behavior of digamma function~\cite{Brychkov}
\begin{equation}
\psi_{0}(x)=\ln x-\frac{1}{2x}+\mathcal{O}\left(\frac{1}{x^2}\right),
\end{equation}
the obtained exact average entropy formula~(\ref{eq:m}) reveals its richer finite-size structure, where polygamma functions up to $m$-th order are involved. 

\begin{figure}[!h]
\begin{center}
\hspace{0cm}\includegraphics[width=1.0\linewidth]{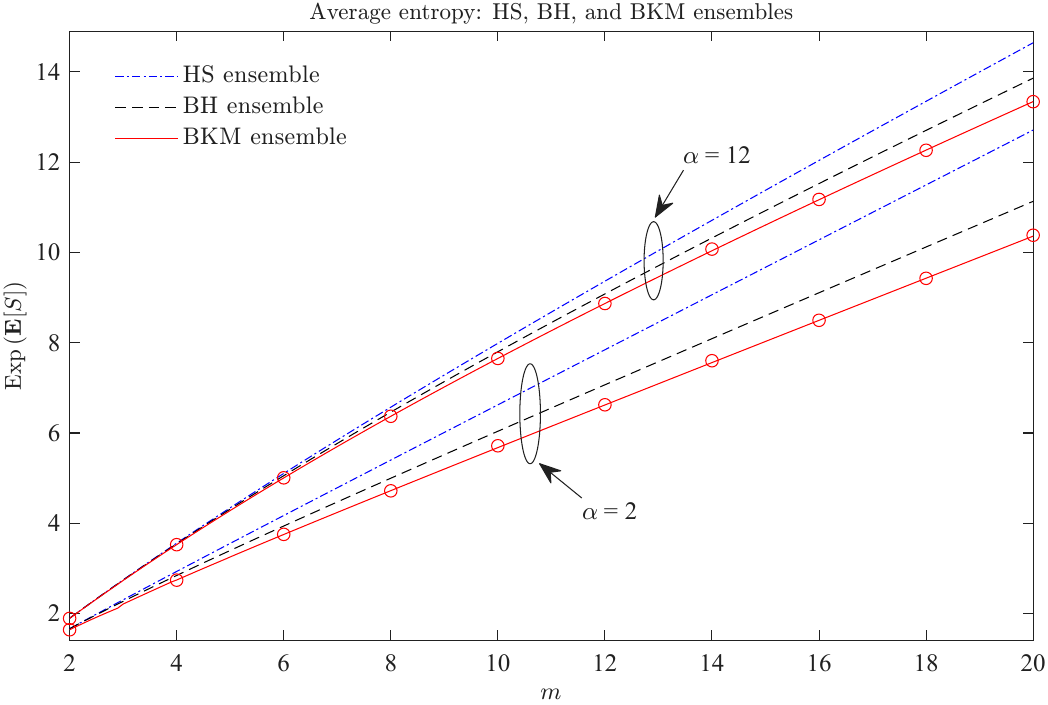}\hspace{0cm}
\caption{Average entropy comparison: the dash-dot lines, dash lines, and solid lines are respectively drawn using exact formulas of HS ensemble in~(\ref{eq:HSm}), BH ensemble in~(\ref{eq:BHm}), and BKM ensemble in~(\ref{eq:m}). The circles represent numerical simulations of BKM ensemble.}
\label{fig:1}
\end{center}
\vspace{0cm}
\end{figure}

As a useful property in quantum information processing, the BKM ensemble~(\ref{eq:fBKM}) is shown in~\cite{Miller25} to attain minimal asymptotic average entropy when compared to Hilbert-Schmidt and Bures-Hall ensembles. It is an open question whether the minimal average entropy property also holds when the dimension is small. The exact formula~(\ref{eq:m}) confirms the validity of this property for arbitrary system dimensions. As illustrated in Figure~\ref{fig:1}, using finite-size formulas~(\ref{eq:HSm}),~(\ref{eq:BHm}), and~(\ref{eq:m}), the average entropy of BKM ensemble is indeed the smallest among three considered ensembles for different values of parameter $\alpha$ and system dimension $m$. By the standard interpretation of von Neumann entropy~(\ref{eq:vN}) supported in $S\in[0,\ln m]$ with $S=0$ for pure states and $S=\ln m$ for maximally mixed states, the BKM ensemble therefore generates the least mixed states on average, the HS ensemble being the most mixed, and the BH ensemble interpolates between the two. Moreover, the gap among the three ensembles becomes more evident for large $m$, which diminishes for small $m$.

It is also seen from Figure~\ref{fig:1} that all three ensembles exhibit an approximate linear increase in the exponential of $\mathbb{E}\left[S\right]$ as dimension $m$ increases, where the scale of y-axis was chosen for improved visualization. In addition, it is observed that the average entropy increases with parameter $\alpha$. Since $\alpha$ is related to the environment (ancilla) dimension, this behavior is consistent with the fact that a larger environment dimension generates more mixed typical entanglement and hence larger average entropy. Last but not least, as a sanity check of the derived formula~(\ref{eq:m}), we perform numerical simulations using log-gas methods~\cite{Kumar19}, represented by circles in Figure~\ref{fig:1}. As expected, the simulation results match well with the analytical formula~(\ref{eq:m}). 

\section{Derivation of average entropy formulas}\label{sec:proof}
In this section, proofs to the main results are provided. Preparatory results required for the later proofs are presented first as summarized in the three lemmas in Section~\ref{sec:aux}. Subsequently, Proposition~\ref{main:k} and Corollary~\ref{main:m} are proved in Section~\ref{sec:P1} and Section~\ref{sec:C1}, respectively.

\subsection{Technical preliminaries}\label{sec:aux}
As a standard first step to facilitate entropic moment calculation~\cite{Foong94,HWC21,HW25,Page93,Ruiz95,Kumar19,VPO16,Wei17,Wei20,Wei20BHA,Wei20BH,Wei23,Wei26,WHW25,WHW26}, one converts the average entropy over the BKM ensemble~(\ref{eq:fBKM}) to that of an ensemble without the delta function constraint
\begin{equation}\label{eq:BKM}
g(x)=\frac{1}{C(\alpha)}\prod_{1\leq i<j\leq m}\left(x_{i}-x_{j}\right)\left(\ln x_{i}-\ln x_{j}\right)\prod_{i=1}^{m}x_{i}^{\alpha}\ee^{-x_{i}}
\end{equation}
supported in $0\leq x_{m}<\dots<x_{1}<\infty$. The relation between average of the induced entropy
\begin{equation}\label{eq:TvN}
T=\sum_{i=1}^{m}x_{i}\ln x_{i}
\end{equation}
over the unconstrained BKM ensemble~(\ref{eq:BKM}) and average of the entropy~(\ref{eq:vN}) over the BKM ensemble~(\ref{eq:fBKM}) is presented in the lemma below.
\begin{lemma}\label{l:f}
\begin{equation}\label{eq:ST}
\mathbb{E}\left[S\right]=\psi_{0}(\beta+1)-\frac{1}{\beta}\mathbb{E}\left[T\right].
\end{equation}
\end{lemma}

\begin{proof}
Under the change of variables
\begin{equation}\label{eq:cv}
\lambda_i=\frac{x_i}{R},~~~~~~i=1,\dots,m
\end{equation}
with
\begin{equation}\label{eq:Rd}
R=\sum_{i=1}^{m}x_i
\end{equation}
denoting the trace, the unconstrained ensemble~(\ref{eq:BKM}) is factorized as
\begin{equation}\label{eq:f}
g(x)\prod_{i=1}^{m}\dd x_i=h(R)f(\lambda)\dd R\prod_{i=1}^m\dd\lambda_i,
\end{equation}
where
\begin{equation}\label{eq:R}
h(R)=\frac{1}{\Gamma(\beta)}R^{\beta-1}\ee^{-R},~~~~~~0\leq R<\infty
\end{equation}
is density of the trace~(\ref{eq:Rd}) with $\beta$ defined in~(\ref{eq:b}).

The factorization~(\ref{eq:f}) implies that $R$ is independent of $\lambda$, and is hence independent of $S$. To exploit this fact, one writes $T$ by the change of variables~(\ref{eq:cv}) as
\begin{equation}
T=R\ln R-RS,
\end{equation}
and consequently
\begin{equation}\label{eq:TS}
\mathbb{E}[T]=\mathbb{E}[R\ln R]-\mathbb{E}[R]\mathbb{E}[S].
\end{equation}
Inserting the computed integrals over~(\ref{eq:R}),
\begin{eqnarray}
\mathbb{E}[R] &=& \beta \\
\mathbb{E}[R\ln R] &=& \beta\psi_{0}(\beta+1) \label{eq:RT}
\end{eqnarray}
into~(\ref{eq:TS}) completes the proof of Lemma~\ref{l:f}.
\end{proof}

Comparing with the result in~\cite[Eq.(23)]{Wei20BH}, it is interesting to observe that the two unconstrained ensembles~(\ref{eq:BKM}) and~(\ref{eq:BH}), despite being structurally distinct, share the same trace distribution~(\ref{eq:R}).

By Lemma~\ref{l:f}, one now works with the unconstrained ensemble~(\ref{eq:BKM}) to obtain $\mathbb{E}[T]$ by first computing $\mathbb{E}\left[R_{l}\right]$. Key ingredients that lead to the expression of $\mathbb{E}\left[R_{l}\right]$ in~(\ref{eq:k}) are properties of the normalization constant $C(\alpha)$,
the results are summarized in the next two lemmas. In Lemma~\ref{l:C} below, we calculate the normalization constant $C(\alpha)$ for any $\alpha>-1$ by exploring its Wronskian structure, where the special case $C(-1/2)$ was obtained in an inductive manner in~\cite{Miller25}. The constant $C(\alpha)$ can also be obtained by the connection of BKM ensemble to Muttalib-Borodin ensemble~\cite{Forrester17,Cheliotis18}.
\begin{lemma}\label{l:C}
\begin{equation}\label{eq:lC}
C(\alpha)=\Gamma^{m}(\alpha+1)\prod_{i=1}^{m}\Gamma(i).
\end{equation}
\end{lemma}

\begin{proof}
Applying Andr\'{e}ief's identity~\cite{Forrester} over the ensemble~(\ref{eq:BKM}), one has
\begin{eqnarray}
C(\alpha) &=&\int_{0}^{\infty}\!\!\!\!\cdots\!\int_{0}^{\infty}\det\left(x_{i}^{j-1}\right)\det\left(\ln^{i-1}\!x_{j}\right)\prod_{i=1}^{m}x_{i}^{\alpha} \ee^{-x_{i}}\dd x_{i} \\
&=& \det\left(\int_{0}^{\infty}\!x^{\alpha+j-1}\ln^{i-1}\!x~\!\ee^{-x}\dd x\right)_{i,j=1}^{m} \\
&=& \det\left(\Gamma^{(i-1)}(\alpha+j)\right)_{i,j=1}^{m}, \label{eq:C1}
\end{eqnarray}
where the determinant~(\ref{eq:C1}) is a Wronskian~\cite{Karlin}
\begin{equation}\label{eq}
W_{m}\left(f_{1},f_{2},\dots,f_{m}\right)=\det
\begin{pmatrix}
f_{1} & f_{1}^{(1)} & \cdots & f_{1}^{(m-1)} \\
f_{2} & f_{2}^{(1)} & \cdots & f_{2}^{(m-1)} \\
\vdots & \vdots & \ddots & \vdots \\
f_{m} & f_{m}^{(1)} & \cdots & f_{m}^{(m-1)}
\end{pmatrix}	
\end{equation}
with
\begin{equation}
f_{i}=\Gamma(\alpha+i).
\end{equation}

Using the Wronskian identity~\cite[Eq.(2.5)]{Karlin}
\begin{equation}
W_{m}\left(f_{1},f_{2},\dots,f_{m}\right)=f_{1}^{m}W_{m-1}\left(\left(\frac{f_{2}}{f_{1}}\right)^{(1)},\left(\frac{f_{3}}{f_{1}}\right)^{(1)},\dots,\left(\frac{f_{m}}{f_{1}}\right)^{(1)}\right),
\end{equation}
we now have
\begin{eqnarray}
C(\alpha) &=& \Gamma^{m}(\alpha+1)W_{m-1}\left((\alpha+1)^{(1)},(\alpha+1)^{(1)}_{2},\dots,(\alpha+1)^{(1)}_{m-1}\right) \\
&=& \Gamma^{m}(\alpha+1)\det
\begin{pmatrix}
(\alpha+1)^{(1)} & (\alpha+1)^{(2)} & \cdots & (\alpha+1)^{(m-1)} \\
(\alpha+1)^{(1)}_{2} & (\alpha+1)^{(2)}_{2} & \cdots & (\alpha+1)^{(m-1)}_{2} \\
\vdots & \vdots & \ddots & \vdots \\
(\alpha+1)^{(1)}_{m-1} & (\alpha+1)^{(2)}_{m-1} & \cdots & (\alpha+1)^{(m-1)}_{m-1}
\end{pmatrix},\label{eq:ci}	
\end{eqnarray}
where
\begin{equation}
(\alpha)_{i}=\frac{\Gamma(\alpha+i)}{\Gamma(\alpha)}
\end{equation}
is the Pochhammer's symbol. Observing that the upper triangular entries in the matrix~(\ref{eq:ci}) are zeros by the fact that for $j>i$,
\begin{equation}
(\alpha+1)^{(j)}_{i}=0,
\end{equation}
whereas the $i$-th, $i=1,\dots,m-1$, diagonal entry is
\begin{equation}
(\alpha+1)_{i}^{(i)}=\left((\alpha+1)(\alpha+2)\dots(\alpha+i)\right)^{(i)}=i!.
\end{equation}
This leads to
\begin{equation}
C(\alpha)=\Gamma^{m}(\alpha+1)\prod_{i=1}^{m-1}i!=\Gamma^{m}(\alpha+1)\prod_{i=1}^{m}\Gamma(i),
\end{equation}
which is the claimed result~(\ref{eq:lC}).
\end{proof}

In the following lemma, we now state some properties of the matrix
\begin{equation}\label{eq:A}
A(\alpha)=\left(\Gamma^{(i-1)}(\alpha+j)\right)_{i,j=1}^{m}
\end{equation}
underlying the determinantal representation~(\ref{eq:C1}) of normalization constant
\begin{equation}\label{eq:C1r}
C(\alpha)=\det\left(A(\alpha)\right).
\end{equation}
These properties are crucial to establish the main results.

\begin{lemma}\label{l:de}
The matrix $A(\alpha)$ admits a decomposition
\begin{equation}\label{eq:LU}
A(\alpha)=A_{L}(\alpha)A_{U}(\alpha)
\end{equation}
into lower and upper triangular matrices
\begin{equation}\label{eq:L}
A_{L}(\alpha)=\left(\binom{i-1}{k-1}\Gamma^{(i-k)}(\alpha+1)\right)_{i,k=1}^{m}
\end{equation}
and
\begin{equation}
A_{U}(\alpha)=\left((\alpha+1)_{j-1}^{(k-1)}\right)_{k,j=1}^{m},
\end{equation}
respectively. These triangular matrices further satisfy
\begin{eqnarray}
A_{L}^{-1}(\alpha)A_{L}(\alpha+l) &=& \left(\binom{i-1}{k-1}\left(\frac{\Gamma(\alpha+l+1)}{\Gamma(\alpha+1)}\right)^{(i-k)}\right)_{i,k=1}^{m}, \label{eq:LL}  \\
A_{U}(\alpha+l)A_{U}^{-1}(\alpha) &=& \left(\frac{l^{j-k}}{(j-k)!}\right)_{k,j=1}^{m}, \label{eq:UU}
\end{eqnarray}
where $A^{-1}$ denotes inverse of a matrix $A$.
\end{lemma}

Before proving Lemma~\ref{l:de}, we note that another LU decomposition of matrix $A(\alpha)$ that involves Stirling numbers was proposed in~\cite{Cheliotis18} for computing normalization constant $C(\alpha)$ of $\alpha=0$. Relationship between the two LU decompositions remains unclear. 

\begin{proof}
Using the definition
\begin{equation}
\Gamma(\alpha+j)=\Gamma(\alpha+1)(\alpha+1)_{j-1},
\end{equation}
the Leibniz rule gives
\begin{eqnarray}
\Gamma^{(i-1)}(\alpha+j) &=& \left(\Gamma(\alpha+1)(\alpha+1)_{j-1}\right)^{(i-1)} \\
&=& \sum_{k=0}^{i-1}\binom{i-1}{k}\Gamma^{(i-k-1)}(\alpha+1)(\alpha+1)_{j-1}^{(k)} \\
&=& \sum_{k=1}^{m}\binom{i-1}{k-1}\Gamma^{(i-k)}(\alpha+1)(\alpha+1)_{j-1}^{(k-1)} \\
&=& \sum_{k=1}^{m}\left(A_{L}(\alpha)\right)_{i,k}\left(A_{U}(\alpha)\right)_{k,j}=\left(A_{L}(\alpha)A_{U}(\alpha)\right)_{i,j},
\end{eqnarray}
where we have utilized the fact that for $k>i$,
\begin{equation}\label{eq:bi0}
\binom{i}{k}=0.
\end{equation}
This establishes the LU decomposition~(\ref{eq:LU}).

To show~(\ref{eq:LL}), we first prove the identity
\begin{equation}\label{eq:LLi}
\left(A_{L}^{f}(\alpha)A_{L}^{g}(\alpha)\right)_{i,j}=\binom{i-1}{j-1}\left(f(\alpha)g(\alpha)\right)^{(i-j)}
\end{equation}
of a more generalized lower triangular matrix than~(\ref{eq:L}) of an arbitrary function $f(\alpha)$,
\begin{equation}
A_{L}^{f}(\alpha)=\left(\binom{i-1}{k-1}f^{(i-k)}(\alpha)\right)_{i,k=1}^{m},
\end{equation}
as
\begin{eqnarray}
\left(A_{L}^{f}(\alpha)A_{L}^{g}(\alpha)\right)_{i,j} &=& \sum_{k=1}^{m}\left(A_{L}^{f}(\alpha)\right)_{i,k}\left(A_{L}^{g}(\alpha)\right)_{k,j} \\
&=& \sum_{k=1}^{m}\binom{i-1}{k-1}\binom{k-1}{j-1}f^{(i-k)}(\alpha)g^{(k-j)}(\alpha) \\
&=& \binom{i-1}{j-1}\sum_{k=j}^{i}\binom{i-j}{k-j}f^{(i-k)}(\alpha)g^{(k-j)}(\alpha) \label{eq:LL1} \\
&=& \binom{i-1}{j-1}\sum_{k=0}^{i-j}\binom{i-j}{k}f^{(i-j-k)}(\alpha)g^{(k)}(\alpha) \\
&=& \binom{i-1}{j-1}\left(f(\alpha)g(\alpha)\right)^{(i-j)},
\end{eqnarray}
where in obtaining~(\ref{eq:LL1}) we have used~(\ref{eq:bi0}) and the identity
\begin{equation}
\binom{i}{k}\binom{k}{j}=\binom{i}{j}\binom{i-j}{k-j}.
\end{equation}

Taking
\begin{equation}
f(\alpha)=\Gamma(\alpha+1),~~~~~~g(\alpha)=\frac{1}{\Gamma(\alpha+1)}
\end{equation}
in~(\ref{eq:LLi}), one has
\begin{equation}
\left(A_{L}^{f}(\alpha)A_{L}^{g}(\alpha)\right)_{i,j}=\left(A_{L}(\alpha)A_{L}^{g}(\alpha)\right)_{i,j}=\delta_{i,j},
\end{equation}
which gives rise to the inverse of $A_{L}(\alpha)$ as
\begin{equation}\label{eq:inL}
A_{L}^{-1}(\alpha)=A_{L}^{g}(\alpha)=\left(\binom{i-1}{k-1}\left(\frac{1}{\Gamma(\alpha+1)}\right)^{(i-k)}\right)_{i,k=1}^{m}.
\end{equation}
The desired identity~(\ref{eq:LL}) is now established by choosing
\begin{equation}
f(\alpha)=\frac{1}{\Gamma(\alpha+1)},~~~~~~g(\alpha)=\Gamma(\alpha+l+1)
\end{equation}
in~(\ref{eq:LLi}) as
\begin{eqnarray}
\left(A_{L}^{f}(\alpha)A_{L}^{g}(\alpha)\right)_{i,k} &=& \left(A_{L}^{-1}(\alpha)A_{L}(\alpha+l+1)\right)_{i,k} \\
&=& \binom{i-1}{k-1}\left(\frac{\Gamma(\alpha+l+1)}{\Gamma(\alpha+1)}\right)^{(i-k)}.
\end{eqnarray}

To show~(\ref{eq:UU}), we consider the shift operator
\begin{equation}
f(\alpha+l)=\ee^{l\frac{\dd}{\dd\alpha}}f(\alpha)=\sum_{k=0}^{\infty}\frac{l^{k}}{k!}f^{(k)}(\alpha),
\end{equation}
and by choosing
\begin{equation}
f(\alpha)=(\alpha+1)^{(i-1)}_{j-1}
\end{equation}
one has
\begin{eqnarray}
(\alpha+l+1)^{(i-1)}_{j-1} &=& \sum_{k=0}^{\infty}\frac{l^{k}}{k!}(\alpha+1)^{(k+i-1)}_{j-1} \\
&=& \sum_{k=i}^{j}\frac{l^{k-i}}{(k-i)!}(\alpha+1)^{(k-1)}_{j-1} \\
&=& \sum_{k=1}^{m}\frac{l^{k-i}}{(k-i)!}(\alpha+1)^{(k-1)}_{j-1}. \label{eq:UU1}
\end{eqnarray}
In a matrix form, the above result~(\ref{eq:UU1}) is
\begin{equation}
A_{U}(\alpha+l)=\left(\frac{l^{k-i}}{(k-i)!}\right)_{i,k=1}^{m}A_{U}(\alpha),
\end{equation}
where right multiplying $A^{-1}_{U}(\alpha)$ on both sides establishes~(\ref{eq:UU}). This completes the proof of Lemma~\ref{l:de}.
\end{proof}

With the above preparations, we are now ready to prove Proposition~\ref{main:k}.

\subsection{Proof of Proposition~\ref{main:k}}\label{sec:P1}
\begin{proof}
Consider the generating function of spectral moments~(\ref{eq:Rk}) of unconstrained BKM ensemble~(\ref{eq:BKM}),
\begin{eqnarray}
\mathbb{E}\left[\ee^{sR_{l}}\right] &=& \frac{1}{C(\alpha)}\int_{0}^{\infty}\!\!\!\!\cdots\!\int_{0}^{\infty}\det\left(x_{i}^{j-1}\right)\det\left(\ln^{i-1}\!x_{j}\right)\prod_{i=1}^{m}x_{i}^{\alpha} \ee^{-x_{i}+sx_{i}^{l}}\dd x_{i} \\
&=& \frac{1}{C(\alpha)}\det\left(\int_{0}^{\infty}\!x^{\alpha+j-1}\ln^{i-1}\!x~\!\ee^{-x+sx^{l}}\dd x\right)_{i,j=1}^{m} \\
&=& \frac{1}{C(\alpha)}\det\left(\frac{\dd^{i-1}}{\dd\alpha^{i-1}}\int_{0}^{\infty}\!x^{\alpha+j-1}\ee^{-x+sx^{l}}\dd x\right)_{i,j=1}^{m},
\end{eqnarray}
the average spectral moments is generated by using multi-linear property of determinant~\cite{Forrester} as
\begin{eqnarray}
\!\!\!\!\!\!\!\!\mathbb{E}\left[R_{l}\right] &=& \frac{\dd}{\dd s}\mathbb{E}\left[\ee^{sR_{l}}\right]\!\bigg|_{s=0} \\
\!\!\!\!\!\!\!\! &=& \frac{1}{C(\alpha)}\sum_{k=1}^{m}\det
\begin{pmatrix}
\Gamma(\alpha+1) & \cdots & \Gamma(\alpha+k+l) &  \cdots & \Gamma(\alpha+m) \\
\Gamma^{(1)}(\alpha+1) & \cdots & \Gamma^{(1)}(\alpha+k+l) & \cdots & \Gamma^{(1)}(\alpha+m) \\
\vdots & \ddots & \vdots & \ddots & \vdots \\
\Gamma^{(m-1)}(\alpha+1) & \cdots & \Gamma^{(m-1)}(\alpha+k+l) & \cdots & \Gamma^{(m-1)}(\alpha+m) \label{eq:Rg}
\end{pmatrix}.
\end{eqnarray}

The sum over $k$ in~(\ref{eq:Rg}) is a sum of $m$ determinants of matrix $A(\alpha)$ defined in~(\ref{eq:A}) with its $k$-th column replaced by $k$-th column of matrix $A(\alpha+l)$. This fact allows us to rewrite~(\ref{eq:Rg}) as
\begin{eqnarray}
\mathbb{E}\left[R_{l}\right] &=& \frac{1}{C(\alpha)}\frac{\dd}{\dd t}\det\left(A(\alpha)+tA(\alpha+l)\right)\!\big|_{t=0} \\
&=& \frac{1}{C(\alpha)}\det\left(A(\alpha)\right)\tr\left(A^{-1}(\alpha)A(\alpha+l)\right) \\
&=& \tr\left(A^{-1}(\alpha)A(\alpha+l)\right),
\end{eqnarray}
where we have used the fact~(\ref{eq:C1r}) and Jacobi's formula of derivative of a determinant~\cite{Forrester}. We now proceed by utilizing the LU decomposition~(\ref{eq:LU}) and associated identities~(\ref{eq:LL}),~(\ref{eq:UU}) in Lemma~\ref{l:de} as
\begin{eqnarray}
\mathbb{E}\left[R_{l}\right] &=& \tr\left(A_{L}^{-1}(\alpha)A_{L}(\alpha+l)A_{U}(\alpha+l)A_{U}^{-1}(\alpha)\right) \\
&=& \sum_{i=1}^{m}\sum_{k=1}^{m}\left(A_{L}^{-1}(\alpha)A_{L}(\alpha+l)\right)_{i,k}\left(A_{U}(\alpha+l)A_{U}^{-1}(\alpha)\right)_{k,i} \\
&=& \sum_{i=1}^{m}\sum_{k=1}^{m}\binom{i-1}{k-1}\left(\frac{\Gamma(\alpha+l+1)}{\Gamma(\alpha+1)}\right)^{(i-k)}\frac{l^{i-k}}{(i-k)!} \\
&=& \sum_{k=0}^{m-1}\sum_{i=0}^{m-1-k}\binom{k+i}{k}\left(\frac{\Gamma(\alpha+l+1)}{\Gamma(\alpha+1)}\right)^{(i)}\frac{l^{i}}{i!} \\
&=& \sum_{i=0}^{m-1}\binom{m}{i+1}\frac{l^{i}}{i!}\left(\frac{\Gamma(\alpha+l+1)}{\Gamma(\alpha+1)}\right)^{(i)},
\end{eqnarray}
where the last equality is obtained by Chu-Vandermonde identity~\cite[Eq.(A10)]{Wei17}
\begin{equation}
\sum_{k=0}^{m-1-i}\binom{k+i}{i}=\binom{m}{i+1}.
\end{equation}
The completes the proof of Proposition~\ref{main:k}.
\end{proof}

With results obtained thus far, it is now a straightforward task to prove Corollary~\ref{main:m}.

\subsection{Proof of Corollary~\ref{main:m}}\label{sec:C1}
\begin{proof}
The average induced entropy $\mathbb{E}[T]$ is generated from spectral moments~(\ref{eq:k}) as
\begin{eqnarray}
\mathbb{E}[T] &=& \frac{\dd}{\dd l}\mathbb{E}\left[R_{l}\right]\!\big|_{l=1} \\
&=& \sum_{i=0}^{m-1}\binom{m}{i+1}\frac{1}{i!}\left((\alpha+1)\psi_{0}(\alpha+2)\right)^{(i)} + \sum_{i=0}^{m-1}\binom{m}{i+1}\frac{1}{(i-1)!}\left(\alpha+1\right)^{(i)} \\
&=& \sum_{i=0}^{m-1}\binom{m}{i+1}\frac{1}{i!}\left((\alpha+1)\psi_{0}(\alpha+1)\right)^{(i)} + \binom{m}{1} + \binom{m}{2},
\end{eqnarray}
where we have used the fact that 
\begin{equation}
\psi_{0}(\alpha+2)=\psi_{0}(\alpha+1)+\frac{1}{\alpha+1}.
\end{equation}
The calculation proceeds by using the identity, cf.~(\ref{eq:polygamma}),
\begin{equation}
\left((\alpha+1)\psi_{0}(\alpha+1)\right)^{(i)}=(\alpha+1)\psi_{i}(\alpha+1)+i\psi_{i-1}(\alpha+1)
\end{equation}
that, after shifting index $i-1$ to $i$, leads to
\begin{eqnarray}
\mathbb{E}[T] &=& \sum_{i=0}^{m-1}\left(\binom{m}{i+1}\frac{\alpha+1}{i!}+\binom{m}{i+2}\frac{1}{i!}\right)\psi_{i}(\alpha+1) + \frac{1}{2}m(m+1) \\
&=& \sum_{i=0}^{m-1}\binom{m}{i+1}\frac{1}{i!}\left(\alpha+\frac{m+1}{i+2}\right)\psi_{i}(\alpha+1)+\frac{1}{2}m(m+1).
\end{eqnarray}
Finally, inserting the above $\mathbb{E}[T]$ expression into the moment conversion formula~(\ref{eq:ST}) of Lemma~\ref{l:f} completes the proof of Corollary~\ref{main:m}.
\end{proof}

\backmatter
\bmhead{Acknowledgments} The work of Lu Wei was supported by the U.S. National Science Foundation (2306968) and the U.S. Department of Energy (DE-SC0024631). Sohail was supported by Grant 63209 from the John Templeton Foundation. The opinions expressed in this publication are those of the authors and do not necessarily reflect the views of the John Templeton Foundation. \\

\begin{appendices}
\section{Other generic state ensembles}\label{appst}
In this appendix, we introduce two major random state ensembles and compare their average entropy formulas with that of the BKM ensemble.

\subsection{Hilbert-Schmidt ensemble}
Hilbert-Schmidt ensemble is considered as the baseline model of generic quantum states. Randomness of the states comes from the assumption of Gaussian distributed coefficients. The probability density function of Hilbert-Schmidt ensemble and its unconstrained version, up to normalization constants, are given respectively by~\cite{Page93}
\begin{equation}\label{eq:fHS}
f(\lambda)\propto\delta\left(1-\sum_{i=1}^{m}\lambda_{i}\right)\prod_{1\leq i<j\leq m}\left(\lambda_{i}-\lambda_{j}\right)^{2}\prod_{i=1}^{m}\lambda_{i}^{\alpha}
\end{equation}
and
\begin{equation}\label{eq:HS}
g(x)\propto\prod_{1\leq i<j\leq m}\left(x_{i}-x_{j}\right)^{2}\prod_{i=1}^{m}x_{i}^{\alpha}\ee^{-x_{i}}.
\end{equation}
Comparing the above unconstrained ensemble~(\ref{eq:HS}) to BKM unconstrained ensemble~(\ref{eq:BKM}), the only difference is that the interaction
\begin{equation}
\prod_{1\leq i<j\leq m}\left(x_{i}-x_{j}\right)
\end{equation}
in~(\ref{eq:HS}) is replaced by the interaction 
\begin{equation}
\prod_{1\leq i<j\leq m}\left(\ln x_{i}-\ln x_{j}\right)
\end{equation}
in~(\ref{eq:BKM}). This in fact leads to substantial differences of average entropy formulas as discussed below.

The average von Neumann entropy~(\ref{eq:vN}) over the Hilbert-Schmidt ensemble~(\ref{eq:fHS}) was conjectured in the seminal work~\cite{Page93} as
\begin{equation}\label{eq:HSm}
\mathbb{E}\!\left[S\right]=\psi_{0}\left(m^{2}+m\alpha+1\right)-\psi_{0}\left(m+\alpha+1\right)-\frac{m-1}{2(m+\alpha)},
\end{equation}
which was subsequently proved, among other proofs, in~\cite{Foong94,Ruiz95}. It is observed that the above average entropy~(\ref{eq:HSm}) only involves digamma function $\psi_{0}$, whereas the average entropy~(\ref{eq:m}) of BKM ensemble of dimension $m$ involves polygamma functions $\psi_{0},\psi_{1},\dots,\psi_{m}$.

Beyond the mean formula~(\ref{eq:HSm}), higher-order cumulants up to the fourth were reported in~\cite{VPO16,Wei17,Wei20,HWC21}, whereas a simplified procedure to derive cumulants of any order was recently proposed in~\cite{HW25}. Finding structures underlying higher-order cumulants of BKM ensemble remains an interesting open question.

\subsection{Bures-Hall ensemble}
Bures-Hall ensemble is an improved variant of the Hilbert-Schmidt ensemble that satisfies a few additional information-theoretic properties~\cite{Sommers03}. The ensemble is often used as a prior distribution known as Bures prior in reconstructing quantum states from measurements. The probability density function of Bures-Hall ensemble and its unconstrained version, up to normalization constants, are given respectively by~\cite{Kumar19}
\begin{equation}\label{eq:fBH}
f(\lambda)\propto\delta\left(1-\sum_{i=1}^{m}\lambda_{i}\right)\prod_{1\leq i<j\leq m}\frac{\left(\lambda_{i}-\lambda_{j}\right)^{2}}{\lambda_{i}+\lambda_{j}}\prod_{i=1}^{m}\lambda_{i}^{\alpha}
\end{equation}
and
\begin{equation}\label{eq:BH}
g(x)\propto\prod_{1\leq i<j\leq m}\frac{\left(x_{i}-x_{j}\right)^{2}}{x_{i}+x_{j}}\prod_{i=1}^{m}x_{i}^{\alpha}\ee^{-x_{i}}.
\end{equation}

The average von Neumann entropy~(\ref{eq:vN}) over the Bures-Hall ensemble~(\ref{eq:fBH}) was conjectured in~\cite{Kumar19} as
\begin{equation}\label{eq:BHm}
\mathbb{E}\!\left[S\right]=\psi_{0}\left(\frac{m(m+1)}{2}+m\alpha+1\right)-\psi_{0}\left(m+\alpha+1\right),
\end{equation}
which was first proved in~\cite{Wei20BHA} and a recent proof based on spectral moments was found in~\cite{WHW26}. It is seen that the above formula~(\ref{eq:BHm}) shares the same first term 
\begin{equation}
\psi_{0}(\beta+1)=\psi_{0}\left(\frac{m(m+1)}{2}+m\alpha+1\right)
\end{equation}
with that of the BKM ensemble~(\ref{eq:m}). This is due to the fact that trace of the two unconstrained ensembles~(\ref{eq:BH}) and~(\ref{eq:BKM}) follows a gamma density~(\ref{eq:R}) with the same shape parameter $\beta$ in~(\ref{eq:b}), where the first terms in~(\ref{eq:BHm}) and~(\ref{eq:m}) are produced from the same trace statistics~(\ref{eq:RT}). It is also of interest to point out that the average entropy formulas~(\ref{eq:BHm}) and~(\ref{eq:HSm}) share the same term
\begin{equation}
\psi_{0}\left(m+\alpha+1\right),
\end{equation}
although the reason for this remains unclear.

Beyond the mean formula~(\ref{eq:BHm}), higher-order cumulants of variance and skewness formulas were obtained in~\cite{Wei20BH} and~\cite{WHW25}, respectively. However, finding cumulant structures of Bures-Hall ensemble that enable the derivation of cumulant of any order is still an open problem. Finally, we note that besides entanglement entropy~(\ref{eq:vN}), other entanglement measures such as entanglement capacity~\cite{Wei23} and relative entropy~\cite{Wei26} have also been studied for the Hilbert-Schmidt and Bures-Hall ensembles.
\end{appendices}\\




\begin{thebibliography}{99}
\bibitem{Balian86}
Balian, R., Alhassid, Y., Reinhardt, H.: Dissipation in many-body systems: A geometric approach based on information theory. Phys. Rep. {\bf 131}, 1 (1986)

\bibitem{Balian14}
Balian, R.: The entropy-based quantum metric. Entropy {\bf 16}, 3878-88 (2014)

\bibitem{Bianchi21}
Bianchi, E., Hackl, L., Kieburg, M.: The Page curve for fermionic Gaussian states. Phys. Rev. B {\bf 103}, L241118 (2021)

\bibitem{Brandao21}
Brand\~ao, F., Chemissany, W., Hunter-Jones, N., Kueng, R., Preskill, J.: Models of quantum complexity growth. PRX Quantum {\bf 2}, 030316 (2021)

\bibitem{Brychkov}
Brychkov, Y.A.: Handbook of Special Functions: Derivatives, Integrals, Series and Other Formulas. CRC Press, Boca Raton (2008)

\bibitem{Cheliotis18}
Cheliotis, D.: Triangular random matrices and biorthogonal ensembles. Stat. Probab. Lett. {\bf 134}, 36-44 (2018)

\bibitem{Choi23}
Choi, J., et al.: Preparing random states and benchmarking with many-body quantum chaos. Nature {\bf 613}, 468 (2023)

\bibitem{Foong94}
Foong, S.K., Kanno, S.: Proof of Page's conjecture on the average entropy of a subsystem. Phys. Rev. Lett. {\bf 72}, 1148-51 (1994)

\bibitem{Forrester}
Forrester, P.: Log-gases and Random Matrices. Princeton University Press, Princeton (2010)

\bibitem{Forrester16}
Forrester, P., Kieburg, M.: Relating the Bures measure to the Cauchy two-matrix model. Commun. Math. Phys. {\bf 342}, 151-87 (2016)

\bibitem{Forrester17}
Forrester, P., Wang, D.: Muttalib-Borodin ensembles in random matrix theory -- realisations and correlation functions. Electron. J. Probab. {\bf 22}, 1-43 (2017)

\bibitem{HWC21}
Huang, Y., Wei, L., Collaku, B.: Kurtosis of von Neumann entanglement entropy. J. Phys. A: Math. Theor. {\bf 54}, 504003 (2021)

\bibitem{HW22}
Huang, Y., Wei, L.: Second-order statistics of fermionic Gaussian states. J. Phys. A: Math. Theor. {\bf 55}, 105201 (2022)

\bibitem{HW23}
Huang, Y., Wei, L.: Entropy fluctuation formulas of fermionic Gaussian states. Ann. Henri Poincar\'{e} {\bf 24}, 4283-342 (2023)

\bibitem{HW25}
Huang, Y., Wei, L.: Cumulant structures of entanglement entropy. arXiv:2502.05371

\bibitem{Karlin}
Karlin, S.: Total Positivity. Stanford University Press, Stanford (1968)

\bibitem{Miller25}
Miller, H.: Entropy-based random quantum states. arXiv:2511.01988v2

\bibitem{Page93}
Page, D.N.: Average entropy of a subsystem. Phys. Rev. Lett. {\bf 71}, 1291-4 (1993)

\bibitem{Ruiz95}
S\'{a}nchez-Ruiz, J.: Simple proof of Page's conjecture on the average entropy of a subsystem. Phys. Rev. E {\bf 52}, 5653-5 (1995)

\bibitem{Kumar19}
Sarkar, A., Kumar, S.: Bures-Hall ensemble: spectral densities and average entropies. J. Phys. A: Math. Theor. {\bf 52}, 295203 (2019)

\bibitem{Sommers03}
Sommers, H.-J., \.{Z}yczkowski, K.: Random Bures mixed states and the distribution of their purity. J. Phys. A: Math. Gen. {\bf 36}, 10083-100 (2003)

\bibitem{VPO16}
Vivo, P., Pato, M.P., Oshanin, G.: Random pure states: Quantifying bipartite entanglement beyond the linear statistics. Phys. Rev. E {\bf 83}, 052106 (2016)

\bibitem{Wei17}
Wei, L.: Proof of Vivo-Pato-Oshanin's conjecture on the fluctuation of von Neumann entropy. Phys. Rev. E {\bf 96}, 022106 (2017)

\bibitem{Wei20}
Wei, L.: Skewness of von Neumann entanglement entropy. J. Phys. A: Math. Theor. {\bf 53}, 075302 (2020)

\bibitem{Wei20BHA}
Wei, L.: Proof of Sarkar-Kumar's conjectures on average entanglement entropies over the Bures-Hall ensemble. J. Phys. A: Math. Theor. {\bf 53}, 235203 (2020)

\bibitem{Wei20BH}
Wei, L.: Exact variance of von Neumann entanglement entropy over the Bures-Hall measure. Phys. Rev. E {\bf 102}, 062128 (2020)

\bibitem{Wei23}
Wei, L.: Average capacity of quantum entanglement. J. Phys. A: Math. Theor. {\bf 56}, 015302 (2023)

\bibitem{Wei26}
Wei, L.: Average relative entropy of random states. J. Phys. A: Math. Theor. {\bf 59}, 215205 (2026)

\bibitem{WHW25}
Wei, L.-F., Huang, Y., Wei, L.: Skewness of von Neumann entropy over Bures-Hall random states. arXiv:2506.06663

\bibitem{WHW26}
Wei, L.-F., Huang, Y., Wei, L.: Spectral moments of Bures-Hall ensemble and applications to entanglement entropy. arXiv:2602.00955

\end{thebibliography}
\end{document}